\documentclass[preprint]{elsarticle}

\usepackage{array,xspace,multirow,hhline,tikz,colortbl,tabularx,booktabs,fixltx2e,amsmath,amssymb,amsfonts,amsthm}
\usepackage{algorithm}
\usepackage{algorithmic}
\usepackage{verbatim,ifthen}
\usepackage{enumitem}
\usepackage{pifont}
\usepackage{calrsfs,mathrsfs}
\usepackage{lscape}
\usepackage{enumitem}
\usetikzlibrary{arrows}
\usetikzlibrary{positioning}
\usepackage{slashbox}

	\newtheorem{lemma}{Lemma}%
	\newtheorem{theorem}{Theorem}%
	\newtheorem{corollary}{Corollary}%
	\newtheorem{example}{Example}

\definecolor{light-gray}{gray}{0.95}
	\newcommand{\pref}{\succ \xspace}

	\usepackage{enumitem}
	\setenumerate[1]{label=\rm(\it{\roman{*}}\rm),ref=({\it\roman{*}}),leftmargin=*}

	\newlength{\wordlength}

	\newcommand{\midd}{\mathbin{:}}

\usepackage{lineno}

\begin{document}

	\title{Efficient Reallocation \\ under Additive and Responsive Preferences}

	%
	%

		\author{Haris Aziz}\ead{haris.aziz@unsw.edu.au}
		\address{UNSW Sydney and Data61, Sydney 2052, Australia}
			\author{P\'{e}ter Bir\'{o}} \ead{biro.peter@krtk.mta.hu}
			\address{Hungarian Academy of Sciences, Budapest, Hungary}
			
					\author{J{\'{e}}r{\^{o}}me Lang} \ead{lang@lamsade.dauphine.fr}
					\author{Julien Lesca} \ead{julien.lesca@dauphine.fr}
	\author{J{\'{e}}r{\^{o}}me Monnot} \ead{jerome.monnot@dauphine.fr}
	\address{LAMSADE, Universite ́ Paris-Dauphine, Paris, France}



\begin{abstract}
Reallocating resources to get mutually beneficial outcomes is a fundamental problem in various multi-agent settings.
While finding an arbitrary Pareto optimal allocation is generally easy, checking whether a particular allocation is Pareto optimal can be much more difficult. This problem is equivalent to checking that the allocated objects cannot be reallocated in such a way that at least one agent prefers her new share to his old one, and no agent prefers her old share to her new one. We consider the problem for two related types of preference relations over sets of objects. In the first part of the paper we focus on the setting in which agents express additive cardinal utilities over objects. We present computational hardness results as well as polynomial-time algorithms for testing Pareto optimality under different restrictions such as two utility values or lexicographic utilities.
In the second part of the paper we assume that agents express only their (ordinal) preferences over single objects, and that their preferences are additively separable.
In this setting, we present characterizations and polynomial-time algorithms for possible and necessary Pareto optimality.
\end{abstract}


\maketitle

\section{Introduction}

Reallocation of resources to achieve  mutually better outcomes is a central concern in multi-agent settings. A desirable way to achieve `better' outcomes is to obtain a Pareto improvement in which each agent is at least as happy and at least one agent is strictly happier.
Pareto improvements are desirable for two fundamental reasons: they result in strictly more welfare for almost any reasonable notion of welfare (such as utilitarian social welfare, or egalitarian `leximin' social welfare). Secondly, they satisfy the minimal requirement of individual rationality in the sense that no agent is worse off after the trade.
If a series of Pareto improvements results in a Pareto optimal outcome, that is even better because there exists  no other outcome which each agent weakly prefers and at least one agent strictly prefers.

We consider the setting in which agents are initially endowed with objects and 
have additive preferences
for the objects. In the absence of endowments, achieving a Pareto optimal assignment is easy in almost all cases: if we know the cardinal utilities that agent have for the objects, we can simply assign every object to the agent who values it the most; if we have only ordinal information, and  assume that an agent always strictly prefer a superset of objects to a subset, then assigning all objects to the same agent leads to a Pareto optimal allocation.\footnote{The only remaining case is where we have ordinal preferences and some agent may be indifferent between receiving an object and not receiving it. In this case, we do not have enough information to say whether an allocation is Pareto optimal or not, nor to find an arbitrary Pareto optimal allocation, independently on computational resources.} 
On the other hand, in the presence of endowments, finding a Pareto optimal assignment that respects individual rationality is more challenging. The problem is closely related to the problem of testing Pareto optimality of the initial assignment: a certificate of Pareto dominance gives an assignment that respects individual rationality and is a Pareto improvement. In fact, if testing Pareto optimality is computationally intractable, then finding an individually rational and Pareto optimal assignment is computationally intractable as well. In view of this, we focus on the problem of testing Pareto optimality. In all cases where we are able to test it
efficiently, we also present algorithms to compute individually rational and Pareto optimal assignments.

We will consider two related settings: one where agents express cardinal preferences, and one where they express ordinal preferences. In the cardinal case, we will assume that not only they express cardinal preferences but {\em additive} preferences, that is, they express a value for every single object, and the value they give to a set of objects is taken to be the sum of the values they gives to each of its elements. We make here a (weak) monotonicity assumption: every value is non-negative. In the ordinal case, we also assume that the underlying preferences of the agents are separable, which leads to assuming that agents have {\em responsive} preferences: an agent prefers to exchange an object $o$ against an object $o'$ in a bundle of objects $S$ (containing $o$) if she prefers to exchange $o$ against $o'$ in any other bundle of objects containing $o$. Because, as we said above, we do now have enough information to test Pareto optimality if some agents may be indifferent between having an object or not, in this setting we will have assume strict monotonicity. 

\paragraph{Contributions}
In the cardinal setting, we first relate the problem of computing an individually rational and Pareto optimal assignment to the more basic problem of testing Pareto optimality of a given assignment.
We show that for an unbounded number of agents, testing Pareto optimality is strongly coNP-complete even if the assignment assigns at most two objects per agent. We then identify some natural tractable cases. In particular,
we present a pseudo-polynomial-time algorithm for the problem when the number of agents is constant.
We characterize Pareto optimality under lexicographic utilities 
and we show that Pareto optimality can be tested in linear time.
For dichotomous utilities in which utilities can take values $\alpha$ or $\beta$, we present a characterization of Pareto optimal assignments which also yields a polynomial-time algorithm to test Pareto optimality.

In the ordinal setting, we consider two versions of Pareto optimality: \emph{possible Pareto optimality} and \emph{necessary Pareto optimality}. For both properties, we present characterizations that lead to polynomial-time algorithms for testing the property for a given assignment.

\paragraph{Related Work}

The setting in which agents express additive cardinal utilities and a welfare maximizing or fair assignment is computed is a very well-studied problem in computer science.
Although computing a utilitarian welfare maximizing assignment is easy, the problem of maximizing egalitarian welfare is NP-hard~\citep{DeHi88a}.


Algorithmic aspects of Pareto optimality have received attention in discrete allocation of indivisible goods~\citep{ACMM05a}, randomized allocation of indivisible goods~\citep{AMXY15a,BoMo01a}, two-sided matching~\citep{CEF+15a,ErEr15a,Manl13a}, coalition formation under ordinal preferences~\citep{ABH11c}, committee elections~\citep{ALL16a}, and probabilistic voting~\citep{ABB14b}.
Since we are interested in Pareto improvements, our paper is also related to housing markets with endowments and ordinal preferences~\citep{AzKe12a,FLS+15a,KQW01a,STSY14a,TSY14a}.
Recently, \citet{DBCM15a} examined restricted Pareto optimality under ordinal preferences. When each agent can get at most one object, there is a well-known characterization of Pareto optimal assignments as not admitting a trading cycle. The result implies that Pareto optimality can be tested in linear time~(see e.g., \citep{ACMM05a,AHR17a}).


De Keijzer et al. \citep{Keijzer09} studied the complexity of deciding whether there exists a Pareto optimal and envy-free assignment when agents have additive utilities. They also showed that testing Pareto optimality under additive utilities is coNP-complete. We show that this result holds even if each agent initially owns two objects.

\citet{CEF+15a} have proved that Pareto optimality of an assignment under lexicographic utilities can be tested in polynomial time. In this paper, we present a simple characterization of Pareto optimality under lexicographic utilities that leads to a linear-time algorithm to test Pareto optimality.

\citet{BEL10a} have considered necessary Pareto optimality as Pareto optimality for all completions of the responsive set extension,\footnote{\citet{BramsEF03} used the term {\em Pareto ensuring} for Pareto optimality for all completions of the responsive set extension.} and have presented some computational results when necessary Pareto optimality is considered \emph{in conjunction}  with other fairness properties. 

Reallocating resources to improve fairness has also been studied before~\citep{Endr13b}.

 \section{Preliminaries}\label{preliminaries}

	
We consider a set of agents $N=\{1,\ldots, n\}$ and a set of objects $O=\{o_1,\ldots, o_m\}$. An assignment $p=(p(1), \ldots , p(n))$ is a partition of $O$ into $n$ subsets, where $p(i)$ is the bundle assigned to agent $i$. We denote by $\mathcal{X}$ the set of all possible assignments. Informally and in examples, we will present an allocation $p$ in the form $(p(1)|\cdots | p(n))$.

In Section \ref{cardinal}, we assume that each agent expresses a cardinal utility function $u_i$ over $O$.
We assume that each object is non-negatively valued, i.e, $u_i(o) \geq 0$ for all $i\in N$ and $o\in O$. 
We also assume additivity: $u_i(O')=\sum_{o\in O'}u_i(o)$ for each $i\in N$ and $O'\subseteq O$. We denote by $\pref_i$ and $\sim_i$ the preference and indifference relations induced by $u_i$, that is, for $S, S' \subseteq O$, $S \pref_i S'$ if $u_i(S) > u_i(S')$, and $S \sim_i S'$ if $u_i(S) = u_i(S')$. If $u_i(o) \geq 0$ for all $o\in O$, we say that $u$ is positively valued. 

In Section \ref{ordinal}, we assume that each agent $i$ expresses only her complete, transitive and reflexive preferences $\pref_i$ over $O$. Agents may be indifferent among objects. Let $\sim_i$ and $\succ_i$ denote the symmetric and anti-symmetric part of $\succsim_i$, respectively. We denote by $E_i^1, \ldots , E_i^{m_i}$ the $m_i$ equivalence classes of an agent $i\in N$. Those classes partition $O$ into $m_i$ sets of objects such that agent $i$ is indifferent between two objects belonging to the same class, and strictly prefers an object of $E_i^k$ to an object of $E_i^l$ whenever $k<l$. The preference profile $\pref=(\pref_1,\ldots, \pref_n)$ specifies for each agent $i$ her preference relation over single objects. 
We will denote by $\mathcal{U}(\pref)$ the set of all utility profiles $u=(u_1,\ldots, u_n)$ such that $u_i\in \mathcal{U}(\pref_i)$ for each $i\in N$.  The way an agent's preference relation is lifted from single objects to sets of objects will be made precise in Section \ref{ordinal}.


{An assignment $p \in \mathcal{X}$ is said to be \emph{individually
  rational} for an initial endowment $e\in \mathcal{X}$ if $u_i(p(i)) \geq u_i(e(i))$ holds for any agent $i$.}
{An assignment $p \in \mathcal{X}$ is said to be \emph{Pareto dominated} by another $q \in \mathcal{X}$ if
 (i) for any agent $i \in N$, $u_i(q(i)) \geq u_i(p(i))$ holds,
 (ii) for at least one agent $i \in N$, $u_i(q(i)) > u_i(p(i))$ holds}. An assignment is \emph{Pareto optimal} if it is not Pareto dominated by another assignment.
Finally, whenever cardinal utilities are considered, the {\em (utilitarian) social welfare} of an assignment $p$ is defined as $SW(p) = \sum_{i\in N}u_i(p(i))$.



\begin{example}\label{example-basic}
Let $n = 3$, $m = 5$, and the utilities of agents be represented as follows.
$$
\begin{array}{c|ccccc}
& o_1 & o_2 & o_3 & o_4 & o_5\\ \hline
1& 16 & 8 & 4 & 2 & 1\\
2 & 10 & 4 & 3 & 3 & 4\\
3 & 6 & 1 & 3 & 6 & 3\\
\end{array}
$$
We assume that the initial assignment is $p=(o_2o_4|o_1|o_3o_5)$ in which $p(1)=\{o_2,o_4\}$, $p(2)=\{o_1\}$, and $p(3)=\{o_3,o_5\}$. We have  $u_1(p(1)) = 12$, $u_2(p(2)) = 10$ and $u_3(p(3)) = 6$. We see that $p$ is not Pareto optimal: indeed, $p' = (o_1|o_2o_3o_5|o_4)$ gives respective utilities 16, 11 and 6 to agents 1, 2 and 3. On the other hand, $p'$ cannot be improved and is Pareto optimal. We can see this by noticing first that agent 1 cannot do better without object $o_1$. Let us give $o_1$ to 1; now, the only way for 2 to do better (without $o_1$) is to have all four remaining objects, in which case 3 gets nothing.  
\end{example}

\section{Additive utilities}\label{cardinal}

In this section we assume that each agent expresses a cardinal
utility function $u_i$ over $O$, where $u_i(o)\geq 0$ for all $i\in N$ and $o\in O$.

\subsection{Testing Pareto optimality: hard cases}

 We will consider Pareto optimality and individual rationality with respect to additive utilities.
 The following lemma shows that the computation of an individually rational and Pareto improving assignment is at least as hard as the problem of deciding whether a given assignment is Pareto optimal:

\begin{lemma}\label{lemma:compute-to-test}
	If there exists a polynomial-time algorithm to compute a Pareto optimal and individually rational assignment, then there exists a polynomial-time algorithm to test Pareto optimality.
	\end{lemma}
	 \begin{proof}
		  We assume that there is a polynomial-time algorithm $A$ to compute an individually rational and Pareto optimal assignment. Consider an assignment $p$ for which Pareto optimality needs to be tested. We can use $A$ to compute an assignment $q$ which is individually rational for the initial endowment $p$ and Pareto optimal. By individual rationality $u_i(q(i))\geq u_i(p(i))$ for all $i\in N$. If $u_i(q(i))= u_i(p(i))$ for all $i\in N$, then $p$ is Pareto optimal simply because $q$ is Pareto optimal. However, if there exists $i\in N$ such that $u_i(q(i))> u_i(p(i))$, it means that $p$ is \emph{not} Pareto optimal.
		\end{proof}

For Example~\ref{example-basic}, assume first that the initial assignment is $p$. Then computing a Pareto optimal and individually rational assignment will return $p'$ (note that for this example there is no other individually rational and Pareto optimal assignment; of course, in general there may be several). Since $p'$ Pareto dominates $p$, we conclude that $p$ is not Pareto optimal. Now, assume that the initial assignment is $p'$. Then computing a Pareto optimal and individually rational assignment will return $p'$, therefore, $p'$ is Pareto optimal. 

A Pareto optimal assignment can be computed trivially by giving each object to the agent who values it the most.\footnote{Note that this particular assignment also maximizes social welfare.} \citet{BoLa14b} proved that a problem concerning coalitional manipulation in sequential allocation is NP-complete (Proposition 6). The result can be reinterpreted as follows.
%

\begin{theorem}
	\label{th:lang}
	Testing Pareto optimality of a given assignment is weakly coNP-complete, even for $n=2$ and identical ordinal preferences.
		\end{theorem}

				\begin{proof}
		We write the proof for the sake of completeness. Testing Pareto optimality is in coNP since one can test that an assignment is Pareto dominated by another one in polynomial time.
		The reduction is from PARTITION. An instance of PARTITION is described by a set of $t$ elements $E=\{e_1,\ldots,e_t\}$, and integer weight $w(e_j)$ for each $e_j\in E$ such that $\sum_{e_i\in E} w(e_i)=2M$. The question is to decide if there exists a balanced partition of $E$ i.e., $S\subseteq E$ such that $\sum_{e_i\in S}w(e_i)=\sum_{e_i\in E\setminus S}w(e_i)=M$?
		The reduction relies on a set of $t+1$ objects $\{g^+,g_1,\ldots, g_t\}$, and two agents $\{ 1, 2\}$.
		The utility values for agent 1 are $u_1(g^+)=M$ and  $u_1(g_i)=w(e_i)$ for all $i\in \{1, \ldots t\}$.
		The utility values for agent 2 are $u_2(g^+)=M+\varepsilon$, with $0<\varepsilon<1$, and $u_2(g_i)=w(e_i)$ for all $i\in \{ 1, \ldots ,t\}$.
		Then, one can easily check that the assignment in which agent $1$ gets $g^+$ and agent $2$ gets all $g_j$ objects is not Pareto optimal if and only if there is a balanced partition of $E$.
					\end{proof}

	\begin{corollary}\label{coro:weaklynphard}
		Computing an individually rational and Pareto optimal assignment is weakly NP-hard for $n=2$.
		\end{corollary}
		%
			

\medskip 

For an unbounded number of agents, testing Pareto optimality of a given assignment is strongly coNP-complete~\citep{Keijzer09}.
We show now that the problem remains strongly coNP-complete, for an unbounded number of agents, even if each agent initially owns exactly two objects.

\begin{theorem}\label{th:testpo-conpc}
	Testing Pareto optimality of a given assignment is strongly coNP-complete, even if each agent receives initially owns two objects.				 	 \end{theorem}
	\begin{proof}
	The reduction is done from {\sc 2-numerical matching with target
	sums} ({\sc 2NMTS} in short). The input of {\sc 2NMTS} is a
	sequence $a_1,\dots, a_k$ of $k$ positive integers such that
	$\sum_{i=1}^k a_i = k(k + 1)$, where $1\leq a_i \leq 2k-1$ for
	$i=1,\dots,k$, and $a_1\leq a_2\ldots \leq a_k$. We want to decide if there are two permutations $\pi$
	and $\theta$ of the integers $\{1,\dots,k\}$ such that $\pi(i)+\theta(i)=a_i$
	for $i=1,\dots,k$. 
	{\sc 2NMTS}  is known to be strongly
NP-complete \cite{YuHL04}. 

The reduction from an instance of 2NMTS is as follows. There are $3k+1$ agents $N=L\cup C\cup R\cup \{d\}$ where $L=\{\ell_1,\dots,\ell_k\}$, $R=\{r_1,\dots,r_k\}$ and $C=\{c_1,\dots,c_k\}$, and $6k+2$ objects $O=F\cup G\cup H\cup\{o\}$ where $F=\{f^L_i,f^R_i:i=1,\dots,k\}$, $G=\{g^L_i,g^R_i:i=1,\dots,k\}\cup \{g^C\}$, and $H=\{h^{CL}_i,h^{CR}_i:i=1,\dots,k\}$. Let $\varepsilon$ be a positive value strictly lower than $1/2$. All utility values are equal to $\varepsilon$ except for the ones summarized in the following table, 
where agt\#1 is the agent which receives the object in the initial assignment and $u_{agt\#1}$ is her utility for it, and where agt(s)\#2 lists the other agents having a utility different from $\varepsilon$ for the object and $u_{agt(s)\#2}$ corresponds to their utility for it:


	
	$$
	\begin{array}{c||c|c||c|c|c}
	\hbox{object}&\hbox{agt}\#1&u_{\hbox{agt}\#1}&\hbox{agt(s)}\#2&u_{\hbox{agt(s)}\#2}\\
	\hline
           h_i^{CL}&c_i&a_i&\ell_i&1+\varepsilon\rule[-3pt]{0pt}{13pt}\\
	\hline
	h_i^{CR}&c_i&3k&r_i&1-\varepsilon\rule[-3pt]{0pt}{13pt}\\
	\hline
	f_i^L&\ell_i&1&c_j ~\hbox{with}~ a_j\geq i+1&i\rule[-3pt]{0pt}{13pt}\\
	\hline
	f_i^R&r_i&1&c_j ~\hbox{with}~ a_j\geq i+1&3k+i\rule[-3pt]{0pt}{13pt}\\
	\hline
	g_i^R&r_i&3&r_{i+1} ~\hbox{if}~ i<k&3+\varepsilon\rule[-3pt]{0pt}{13pt}\\
	&&&d ~\hbox{if}~ i=k&3+\varepsilon\\
	\hline
	g_i^L&\ell_i&3&\ell_{i-1} ~\hbox{if}~ i>1&3-\varepsilon\rule[-3pt]{0pt}{13pt}\\
	&&&r_1 ~\hbox{if}~ i=1&3+\varepsilon\\
	\hline
	g^C&d&3&\ell_k&3-\varepsilon\rule[-3pt]{0pt}{13pt}\\
	\hline
	o&d&1&&\\
           \end{array}
	$$
The initial assignment gives the following utilities to agents: $u_d(\{ g^C, o\})=4$, $u_{c_i}(\{ h_i^{CL}, h_i^{CR}\})=3k+a_i$, $u_{\ell_i}(\{f_i^L, g_i^L \})=4$ and $u_{r_i}(\{f_i^R, g_i^R \})=4$ for $i=1\ldots k$.

\smallskip

	Clearly, this instance is constructed within polynomial time and each agent has initially two objects. We claim that there is a Pareto improvement of the initial assignment if and only if  $\{a_i: i=1\ldots k\}$ is a yes-instance of {\sc 2NMTS}. 

\smallskip

	Assume first that there exist $\pi$ and $\theta$ such that $\pi(i)+\theta(i)=a_i$ for $i=1\ldots k$ i.e., $\{a_i: i=1\ldots k\}$ is a yes-instance of {\sc 2NMTS}. Note that this implies for any $i=1\ldots k$ that
\begin{equation}\label{equationproof}
\pi(i)+1\leq a_i~~\hbox{and}~~\theta(i)+1\leq a_i
\end{equation}
because $\pi(i)\geq 1$ and $\theta(i)\geq 1$. Then consider the following assignment:
\begin{itemize}
\item $\{h_i^{CL}, g_{i+1}^L\}$ (resp. $\{h_k^{CL}, g^C\}$) is assigned to $\ell_i$ with $i<k$ (resp. to $\ell_k$) with utility 4.
\item $\{h_i^{CR}, g_{i-1}^R\}$ (resp. $\{h_1^{CR}, g_1^L\}$) is assigned to $r_i$ with $i>1$ (resp. to $r_1$) with utility 4.
\item $\{f_{\pi(i)}^R, f_{\theta(i)}^L\}$ is assigned to $c_i$. By using (\ref{equationproof}), it is easy to check that the utility of agent $c_i$ is $3k+\pi(i)+\theta(i)=3k+a_i$.
\item $\{o, g_k^R\}$ is assigned to $d$ with utility $4+\varepsilon$.
\end{itemize}
This allocation is clearly a Pareto improvement of the initial allocation.

\smallskip

Assume now that $\{a_i: i=1\ldots k\}$ is a no-instance of {\sc 2NMTS}. By contradiction, assume that there exists a Pareto improvement $p$ of the initial assignment. Note first that any agent should receive in $p$ at least two objects. Indeed, there is no object which provides a utility greater than $3+\varepsilon$ to any agent of $L\cup R\cup \{ d\}$, and any of those agents receives a utility of 4 in the initial assignment. Furthermore, any object $f_i^R$ provides a utility of at most $3k+i$ to agent $c_j$, which is strictly lower than her utility of $3k+a_j$ in the initial assignment because $a_j\geq i+1$ (otherwise $c_j$ would get utility 0 from $f_i^R$). Since the number of objects is twice the number of agents, we can conclude that $p$ assigns exactly 2 objects to every agent.


\begin{figure}
\begin{center}
\includegraphics[clip,height=8.5cm,keepaspectratio]{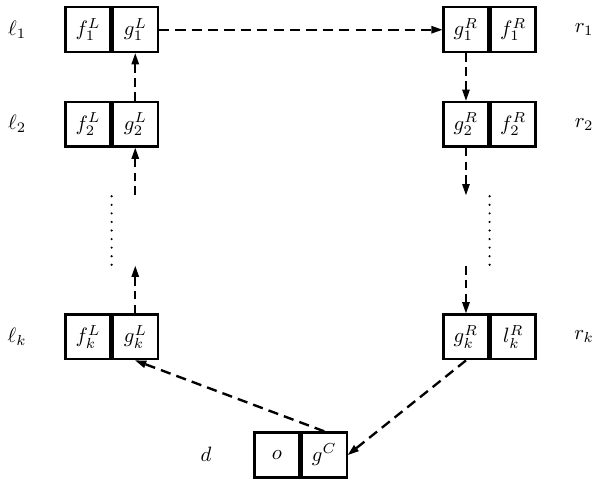}
\end{center}
\caption{Initial assignment for agents of $L\cup R\cup \{ d\}$.}
\label{figCycleProof}
\end{figure}

Let us focus first on the objects of $G$. These objects are the only ones that can provide a utility of at least $3-\varepsilon$ to the agents of $L\cup R\cup \{ d\}$. Any other object provides a utility of at most $1+\varepsilon$ to the agents of $L\cup R\cup \{ d\}$. Therefore, each agent in $L\cup R\cup \{ d\}$ should receive exactly one object from $G$ to achieve a utility of at least 4 since $|L\cup R\cup \{ d\}|=|G|=2k+1$. Figure \ref{figCycleProof} illustrates the initial assignment for the agents of $L\cup R\cup \{ d\}$. In this figure, a dashed-line arrow from an object of $G$ means that the agent pointed by the arrow is the only one that has a utility value different from $\varepsilon$ for the object. Figure \ref{figCycleProof} illustrates the fact that the objects of $G$ could be allocated in only two different manners in $p$ to be a Pareto improvement of the initial endowment: either every object of $G$ is assigned to the same agent as in the initial assignment, or every object of $G$ is assigned to the agent pointed at by the corresponding arrow in Figure \ref{figCycleProof}.

First, we consider the case where all objects of $G$ are assigned in $p$ exactly as in the initial assignment. To achieve a utility of at least 4, every agent $r_i$ should receive object $f_i^R$ to complete her bundle of two objects. This implies that these objects cannot be assigned to agent $c_i$, with $i=1\ldots k$, in order to ensure that they get a utility of at least $3k+a_i$. Therefore, each agent $c_i$ should receive object $h_i^{CR}$ with utility $3k$. Furthermore no agent $c_i$ can receive an object $f_j^L$ to complete her bundle of two objects because this object would provide her a utility of at most $a_i-1$. Hence, each agent $c_i$ should receive object $h_i^{CL}$. All in all, $p$ should be exactly the same assignment as the initial assignment, which contradicts our assumption that $p$ Pareto dominates this initial assignment.

From the previous paragraphs, we know that any object of $G$ should be assigned in $p$ to the agent pointed at by the corresponding dotted arrow in Figure \ref{figCycleProof}. To achieve a utility of at least 4, any agent $\ell_i$ should receive object $h_i^{CL}$ to complete her bundle of two objects. If agent $c_i$ do not receive at least one object $f_j^R$ such that $a_i\geq j+1$, then the maximal utility achievable by $c_i$ would be $3k+a_i-1$, which would be strictly lower than her utility in the initial assignment. Hence, each agent $c_i$ should receive exactly one object $f_j^R$ such that $a_i\geq j+1$. This implies that no object $f_i^R$ can be assigned to agent $r_i$. Hence, to achieve a utility of at least 4, any agent $r_i$ should receive object $h_i^{CR}$ to complete her bundle of two objects. Then object $o$ should be assigned to agent $d$ to complete her bundle of two objects. Finally it remains to assign to each agent $c_i$ object $f_j^L$ such that $a_i\geq j+1$.

Now we focus on the pair of objects assigned to each agent $c_i$ in $p$. Note that these two objects belong to $F$. We know that the total amount of utility provided by the objects of $F$ to the agents of $C$ should be exactly equal to $3k^2+k(k+1)$. Furthermore, any agent $c_i$ should receive a share of at least $3k+a_i$ of this total amount of utility. Since $\sum_{i=1}^k(3k+a_i)=3k^2+k(k+1)$, any agent $c_i$ should receive two objects $f_j^L$ and $f_{j'}^R$ such that $u_{c_i}(\{f_j^L, f_{j'}^R\})=3k+a_i$. Let $\pi$ and $\theta$ be the two permutations of $\{ 1, \ldots , k\}$ such that for any $i=1\ldots k$, objects $f_{\pi(i)}^L$ and $f_{\theta(i)}^R$ are assigned in $p$ to agent $c_i$. These two permutations are such that for any $i=1\ldots k$, $\pi(i)+\theta(i)=a_i$. This leads to a contradiction with our assumption that $\{ a_i: i=1\ldots k\}$ is a no-instance.
\end{proof}

%

Note that  Theorem \ref{th:testpo-conpc} is tight according to the number of objects initially owned by each agent because if initially each agent has exactly one object in assignment $p$, then our problem is well-known to be solvable in linear time~\citep{ACMM05a}.


Our final result show that even if agents utility for objects are bounded to three values then the complexity of testing Pareto optimality remains strongly coNP-complete.

\begin{theorem}\label{th:testpo-three}
	Testing Pareto optimality of a given assignment is strongly coNP-complete,
	even if there are only three utility values.				 	 \end{theorem}
	\begin{proof}
The reduction is from (3,B2)-SAT \cite{Kratochvl1993OneMO} which is a restriction of 3SAT where each literal appears exactly twice in the clauses, and therefore, each variable appears four times. An instance of (3,B2)-SAT is composed by set $X=\{x_1,...,x_s\}$ of $s$ variables and set $C=\{c_1,...,c_t\}$ of $t$ clauses. Each clause in $C$ contains exactly 3 literals. The decision problem is to define whether there exists a truth assignment of $X$ such that all clauses in $C$ are true.

The reduction from an instance of (3,B2)-SAT is as follows. Let $q$ denote the initial endowment. For each variable $x_i$, we create two agents $V_i$ and $W_i$ with $q(V_i)=\{v_i^1, v_i^2\}$ and $q(W_i)=\{w_i^1, w_i^2\}$. Both $V_i$ and $W_i$ have utility 1 for each good in their respective initial endowment, and both of them have utility $t$ for $h_i$.
Furthermore, for each clause $c_j$ we create agent $Z_j$  with $q(Z_j)=\{z_j\}$. Agent $Z_j$ has utility 1 for $z_j$, and she has utility 1 for any good corresponding to a literal of $c_j$ i.e., she has utility 1 for $v_i^1$ and $v_i^2$ ($w_i^1$ and $w_i^2$, respectively) if $x_i$ ($\neg x_i$, respectively) is a literal appearing in $c_j$. Finally, we create two additional agents $A$ and $H$ with $q(A)=\{a\}$ and $q(H)=\{h_1,...,h_t\}$. Agent $A$ has utility $t$ for $a$, and she has utility 1 for each object $z_j$, with $j=1..t$. Agent $H$ has utility 1 for each good in $q(H)$, and she has utility $t$ for $a$. All  the other utility values are equal to $\varepsilon=\frac{1}{6t}$. The following table summarize the utility values which differ from $\varepsilon$, where agt\#1 is the agent which receives the object in the initial assignment and $u_{agt\#1}$ is her utility for it, and where agt(s)\#2 lists the other agents having a utility different from $\varepsilon$ for the object and $u_{agt(s)\#2}$ corresponds to their utility for it:

$$
	\begin{array}{c||c|c||c|c|c}
	\hbox{object}&\hbox{agt}\#1&u_{\hbox{agt}\#1}&\hbox{agt(s)}\#2&u_{\hbox{agt(s)}\#2}\\
	\hline
           v_i^{\ell}&V_i&1&Z_j \hbox{ where } c_j \hbox{ contains }x_j&1\rule[-3pt]{0pt}{13pt}\\
	\hline
	w_i^{\ell}&W_i&1&Z_j \hbox{ where } c_j \hbox{ contains }\neg x_j&1\rule[-3pt]{0pt}{13pt}\\
	\hline
	z_j&Z_j&1&A&t\rule[-3pt]{0pt}{13pt}\\
	\hline
	h_i&H&1&V_i&t\rule[-3pt]{0pt}{13pt}\\
	&&&W_i&t\\
	\hline
	a&A&t&H&t\rule[-3pt]{0pt}{13pt}\\
	\end{array}
$$
We claim that the instance of (3,B2)-SAT is a yes-instance if and only if $q$ is not Pareto optimal.

Assume that there exists truth assignment $\varphi$ of $X$ such that all clauses in $C$ are true. We show that we can construct  from $\varphi$ an assignment $p$ that Pareto dominates $q$. For each variable $x_i$ that is true according to $\varphi$, $p(V_i)=\{ h_i\}$ and $p(W_i)=\{w_i^1, w_i^2\}$. For each variable $x_i$ that is false according to $\varphi$, $p(V_i)=\{v_i^1, v_i^2\}$ and $p(W_i)=\{h_i\}$. Note that in both cases, $u_{V_i}(p(V_i))\geq u_{V_i}(q(V_i))=2$ and $u_{W_i}(p(W_i))\geq u_{W_i}(q(W_i))=2$ hold, and at least one of these inequalities is strict. For each clause $c_j$ and for each literal $x_i$ ($\neg x_i$, respectively) appearing in $c_j$, if $x_i$ is true (false, respectively) according to $\varphi$ then $v_i^{\ell}$ ($w_i^{\ell}$, respectively) is assigned to $Z_j$ in $p$, where $\ell$ equals 1 if the first occurrence of $x_i$ ($\neg x_i$, respectively)  is in $c_j$ and $\ell$ equals 2 otherwise. Since there is at least one literal true according to $\varphi$ for each clause in $C$, at least one object is assigned to each $Z_j$ in $p$ and $u_{Z_j}(p(Z_j))\geq u_{Z_j}(q(Z_j))=1$ holds. Finally, $p(H)=\{ a\}$ and $p(A)=\{z_1, \ldots , z_t\}$. Note that both $u_{A}(p(A))=u_{A}(q(A))=t$ and $u_{H}(p(H))=u_{H}(q(H))=t$ hold. Therefore, $p$ Pareto dominates $q$.

Assume now that there exists an assignment $p$ that Pareto dominates $q$. 
Note first that $p(V_i)=q(V_i)$ and $p(W_i)=q(W_i)$ cannot hold for any $i=1..t$. Indeed, otherwise it is easy to check that each agent should receive in $p$ the same bundle than in $q$ to ensure that each of them receives a bundle at least as good as her bundle in $p$, leading to a contradiction with $p$ Pareto dominates $q$. This means that for some $k\in \{ 1, \ldots , t\}$, either $V_k$ or $W_k$ receive a bundle different from $q$. In both cases, this bundle should contain $h_k$ to ensure that the utility of $V_k$ or $W_k$ does not decrease. But this implies that $h_k$ cannot be assigned to $H$ in $p$. To ensure $u_H(p(H))\geq u_H(q(H))=t$, agent $H$ should receive $a$. Since $a$ is not assigned to $A$, $\{z_1, \ldots , z_t \}\subseteq p(A)$ should hold to ensure $u_A(p(A))\geq u_A(q(A))=t$. Hence, no agent $Z_j$ receives object $z_j$. To achieve an utility greater than or equal to $q(Z_j)=1$, each agent $Z_j$ should receive in $p$ either \textit{(i)} good $v_i^{\ell}$ corresponding to one of its positive literals, or \textit{(ii) }good $w_i^{\ell}$ corresponding to one of its negative literals. In case \textit{(i)}, $v_i^{\ell}$ is not assigned to $V_i$, and $u_{V_i}(p(V_i))\geq u_{V_i}(q(V_i))$ implies $h_i\in p(V_i)$. Therefore, $h_i$ is not assigned to $W_i$, and $p(W_i)=q(W_i)=\{ w_i^1, w_i^2\}$ is necessary to ensure $u_{W_i}(p(W_i))\geq u_{W_i}(q(W_i))$. This means that no agent $Z_{j'}$ receives object $w_i^1$ or $w_i^2$. In case \textit{(ii)}, the same type of reasoning leads to conclude that if agent $Z_j$ receives $w_i^{\ell}$ then no agent $Z_{j'}$ receives object $v_i^1$ or $v_i^2$. Therefore, we can construct from $p$ a truth assignment of $X$ that satisfy all clauses in $C$ by setting the truth value of variable $x_i$ to true if object $v_i^1$ or $v_i^2$ is assigned to some agent $Z_j$, and to false otherwise. In other words, the instance of (3,B2)-SAT is a yes-instance.
\end{proof}

\subsection{Complexity of testing Pareto optimality: tractable cases}	

We now identify conditions under which the problem of computing individually rational and Pareto optimal assignments is polynomial-time solvable.
	
\subsubsection{Constant number of agents and small utilities}

We first focus on the case of a constant number of agents and small utility weights. 
					
						\begin{lemma}\label{lemma:dp}
If there is a constant number of agents, then the set of all vectors of utilities that correspond to an assignment can be computed in pseudo-polynomial-time.
						\end{lemma}

\begin{proof}
Consider the following algorithm (by $0^k$ we denote $0, \ldots, 0$ with $k$ occurrences of $0$).
\begin{algorithm}[h!]
\label{algo1v1}
\begin{algorithmic}[1]
	\footnotesize
\STATE $L \leftarrow \{ (0^n) \}$;
\FOR{$j = 1$ to $m$}
\STATE{$L' \leftarrow \{ l + (0^{i-1}, u_i(o_j), 0^{n-i}) \ | \ i \in N; l \in L\}$}
\STATE{$L \leftarrow L'$}
\ENDFOR
\RETURN $L$
\end{algorithmic}
\end{algorithm}

Let $W$ be the maximal social welfare that is achievable; then, at any step of the algorithm, the number of vectors in $L$ cannot exceed $(W+1)^n$. Hence the algorithm runs in $O(W^n\cdot n \cdot m)$. Now, $W \leq \sum_{i,j} u_i(o_j)$, and since $n$ is constant, the algorithms runs in pseudo-polynomial-time.

We can prove by induction on $k$
that a vector of utilities $l = (v_1, \ldots, v_n)$ can be achieved by assigning objects $o_1, \ldots, o_k$ to the agents if  and only if $l$ belongs to $L$ after objects $o_1, \ldots, o_{k}$
have been considered.
This is obviously true at the start of the algorithm, when no object at all has been considered. Now, suppose the induction assumption is true for $k$. If $l$ belongs to $L$ after iteration $k$, then $l'$ belongs to $L$ after iteration $k+1$ if
$l'$ is obtained from $l$ by adding $u_i(o_k)$ to the utility of some agent $i$, that is, if $l = (v_1, \ldots, v_n)$ can be achieved by assigning objects $o_1, \ldots, o_{k+1}$.
\end{proof}

\begin{theorem}\label{theo:pseudopoly}
If there is a constant number of agents, then there exists a pseudo-polynomial-time algorithm to compute a Pareto optimal and individually rational assignment.
\end{theorem}
								
\begin{proof}
We apply the algorithm of Lemma \ref{lemma:dp}, but in addition we keep track, for each $l \in L$, of a partial assignment that supports it: every time we add $l + (0^{i-1}, u_i(o_j), 0^{n-i})$ to $L'$, the corresponding partial assignment is obtained from the partial assignment corresponding to $l$, and then mapping $o_j$ to $i$. If several partial assignments correspond to the same utility vector, we keep an arbitrary one. At the end, we obtain the list $L$ of all feasible utility vectors, together with, for each of them, one corresponding assignment. For each of them, check whether there is at least one $l'$ in $L$ that Pareto dominates it, which takes at most $O(|L|^2)$, which is polynomially large if $L$ is encoded in unary.
The assignments that correspond to the remaining vectors are Pareto optimal.\footnote{Note that it is generally not the case that we get {\em all} Pareto optimal assignments: if there are several assignments corresponding to the same utility vector, then we'll obtain only one.}
\end{proof}

	\subsubsection{Lexicographic Utilities}

We say that utilities are \emph{lexicographic} if for each agent $i\in N$ and each object $o\in O$, {$u_i(o)>\sum_{o'\prec_i o}u_i(o')$} with the convention $\sum_{o \in \emptyset} u_i(o) = 0$, which implies that $u_i(o) > 0$ for each $o$.
Note that agent 1 in Example~\ref{example-basic} has lexicographic utilities.
So would be an agent with utilities $(8,8,3,3,1)$. 


In order to test the Pareto optimality of an assignment $p$, we construct a directed graph $G(p)=(V(p), E(p))$. The set of vertices $V(p)$ contains one vertex per object belonging to $O$. Furthermore, for any vertex of $V(p)$ associated to object $o$, the set of edges $E(p)$ contains one edge $(o, o')$ for any object $o'$ belonging to $O\setminus \{ o\}$ such that $o'\succsim_i o$, where $i$ is the agent who receives object $o$ in $p$.

	\begin{example}\label{example-lex0}
Let $n = 3$, $m = 5$, and the following  ordinal information about preferences corresponding to the lexicographic utilities. (Due to the lexicographicity assumption, this ordinal information is enough to know completely an agent's preference relation over all sets of objects.)
$$
\begin{array}{c}
1: o_1 \succ o_2 \succ o_3 \succ o_4 \succ o_5\\
2: o_1 \succ o_2 \sim o_5 \succ o_3 \sim o_4\\
3: o_1 \sim o_4 \succ o_3 \sim o_5 \succ o_2
\end{array}
$$
Let $p=(o_2o_4|o_1|o_3o_5)$ be the initial assignment. 
Then, the corresponding graph $G(p)$ is depicted in Figure \ref{figLex}. In Figure~\ref{figLex}, dashed-line edges represent indifferences (when $o'\sim o$) and solid-line edges represent strict preferences (when $o'\succ o$).
	\end{example}


It follows from \citep{CEF+15a} that Pareto optimality of an assignment under lexicographic utilities can be tested in polynomial time. 
We provide a simple characterization of a Pareto optimal assignment under lexicographic utilities. The characterization we present also provides an interesting connection with the possible Pareto optimality that we consider in the next section.

\begin{figure}
	\begin{center}
	\includegraphics[clip,height=4cm,keepaspectratio]{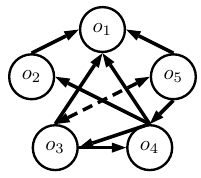}
	\end{center}
	\caption{Graph $G(p)$ for assignment $p$ in Example~\ref{example-lex0}.}
	\label{figLex}
	\end{figure}

\begin{theorem}\label{th:lexi-opt}
An assignment $p$ is not Pareto optimal with respect to lexicographic utilities if and only if there exists a cycle in $G(p)$ which contains at least one edge corresponding to a strict preference.
\end{theorem}
	\begin{proof}
Assume that there exists a cycle $C$ that contains at least one edge corresponding to a strict preference.
Then, the exchange of objects along the cycle by agents owning these objects corresponds to a Pareto improvement.

Assume now that $p$ is not Pareto optimal and let $q_1$ be an assignment that Pareto dominates $p$. For at least one agent $i$, $q_1(i)\succ_{i} p(i)$. Therefore, there exists at least one object $o_1$ in $q_1(i)\setminus p(i)$. Let $i_1$ be the owner of $o_1$ in $p$. Since preferences are lexicographic, in compensation of the loss of $o_1$, agent $i_1$ must receive an object $o_2$ in $q_1$ which is at least as good as $o_1$ according to her own preferences. Let $i_2$ be the owner of $o_2$ in $p$ and so on. Since $O$ is finite, there must exist $k$ and $k'$ such that the sequence $o_{k}\rightarrow o_{k+1}\rightarrow \ldots \rightarrow o_{k'}$ forms a cycle, i.e., $o_k=o_{k'}$. If $\not\exists l \in [k, k'-1]$ such that $o_{l+1}\succ_{i_l} o_l$ then we consider assignment $q_2$ derived from $q_1$ by reassigning any object $o_{l+1}$, with $l\in [k, k'-1]$, to agent $i_l$. It is obvious that assignment $q_2$ is at least as good as $q_1$ for all agents. Hence, $q_2$ Pareto dominates $p$. By following the same reasoning as above, we can state that there exists a sequence of objects $o_{k}\rightarrow o_{k+1}\rightarrow \ldots \rightarrow o_{k'}$ such that $o_k=o_{k'}$ and for any $l\in [k, k'-1]$, $o_{l+1}$ is assigned to agent $i_l$ in $q_2$ to compensate the loss of $o_l$ assigned to him in $p$ (obviously with $o_{l+1}\succsim_{i_l} o_{l}$). Once again, if there is no $l \in [k, k'-1]$ such that $o_{l+1}\succ_{i_l} o_l$ then we consider assignment $q_3$ derived from $q_2$ by reassigning any object $o_{l+1}$, with $l\in [k, k'-1]$, to agent $i_l$, and so on. Since for any $s>1$ we have $\sum_{i\in N}|q_{s-1}(i)\cap p(i)|>\sum_{i\in N}|q_{s}(i)\cap p(i)|$, there must exist a finite value $t$ where $\exists l \in [k, k'-1]$ such that $o_{l+1}\succ_{i_l} o_l$ for the cycle $o_{k}\rightarrow o_{k+1}\rightarrow \ldots \rightarrow o_{k'}$ founded in $q_t$. Indeed otherwise after a finite number of steps $t$ we would have $q_t(i)=p(i)$ for all $i\in N$, leading to a contradiction with our assumption that $q_t$ Pareto dominates $p$. Therefore, we have shown that there exists a cycle $o_k\rightarrow \ldots \rightarrow o_{k'}$ in $G(p)$ with at least one edge corresponding to a strict preference.
\end{proof}	

It is clear that graph $G(p)$ can be constructed in linear time for any assignment $p$. Furthermore, the search of a cycle containing at least one strict preference edge in $G(p)$ can be computed in linear time by applying a graph traversal algorithm for each strict preference edge in $G(p)$. Therefore, testing if a given assignment is Pareto optimal can be done in linear time when utilities are lexicographic.

The construction of Theorem \ref{th:lexi-opt}, and illustrated in Figure \ref{figLex}, gives us that assignment $p$ in Example~\ref{example-lex0}
is Pareto dominated by $(o_2o_3|o_1|o_4o_5)$. 


%

		\subsubsection{Two utility values}

In this section we assume that agents use at most two utility values to evaluate objects. We say that the collection of utility functions $(u_1, \ldots, u_n)$ is {\em bivalued}
if there exist two numbers $\alpha > \beta \geq 0$ such that for every agent $i$ and every object $o$, $u_i(o) \in \{\alpha, \beta\}$.\footnote{Note that the result would still hold if each agent $i$ has a different pair of values $(\alpha_i,\beta_i)$, provided that $\frac{\alpha_i}{\beta_i} =  \frac{\alpha_j}{\beta_j}$ for all $i, j$.} This implies that for each agent $i$, the set of objects $O$ is partitioned into two subsets $E_i^1 = \{o \in O, u_i(o) = \alpha\}$ and $E_i^2 = \{o \in O, u_i(o) = \beta\}$ (with possibly $E_i^2 = \emptyset$). Given an assignment $q$, let $q^+(i) = q(i) \cap E^1_i$, and $q^-(i) = q(i) \cap E^2_i$.

		 We provide a first requirement for an assignment to Pareto dominate another one:

		\begin{lemma}\label{lemm2}
		If an assignment $p$
		is Pareto dominated by an assignment $q$ then $|\bigcup_{i\in N}q^+(i)|>|\bigcup_{i\in N}p^+(i)|$.
		\end{lemma}
		\begin{proof}
		For contradiction we assume that $|\bigcup_{i\in N}q^+(i)|\leq|\bigcup_{i\in N}p^+(i)|$. In that case, $SW(q) = |\bigcup_{i\in N}q^+(i)|\alpha + |\bigcup_{i\in N}q^-(i)|\beta = |\bigcup_{i\in N}q^+(i)|(\alpha-\beta) + |O|\beta \leq |\bigcup_{i\in N}p^+(i)|(\alpha-\beta) + |O|\beta = \alpha|\bigcup_{i\in N}p^+(i)| + \beta |\bigcup_{i\in N}p^-(i)|$.
Therefore, $SW(p) \geq SW(q)$ holds, which contradicts the assumption that $q$ Pareto dominates $p$.
		\end{proof}

		\begin{lemma}\label{lemm1}
		If an assignment $p$
		is not Pareto optimal then there exists an assignment $q$ that
 Pareto dominates $p$ with the two following properties:
 (i) $\forall i\in N, |q^+(i)|\geq |p^+(i)|$ and (ii) $\exists j\in N, |q^+(j)|>|p^+(j)|$ and $p^-(j)\neq \emptyset$.
		\end{lemma}
		
				\begin{proof}
Assume that $p$ is not Pareto optimal. Then there exists an assignment which Pareto dominates $p$. Among the set of assignment that Pareto dominates $p$, let $q_*$ be the one that is the closest to $p$ i.e., $|\bigcup_{i\in N}q_*(i)\setminus p(i)|$ is minimal.

First we note that the above assumption implies that there is no \emph{clear winner} agent $i$ such that $p^-(i)=\emptyset$ and $q_*^+(i)\supset p^+(i)$. Indeed, otherwise we could reallocate each object in $q_*^+(i)\setminus p^+(i)$ to its owner in $p$, and obtain another assignment $q_{**}$ from $q_*$ that also Pareto dominates $p$, but which is closer to $p$ than $q_*$.

Lemma \ref{lemm2} implies that $|\bigcup_{i\in N}q_*^+(i)|>|\bigcup_{i\in N}p^+(i)|$ holds. Let $o_1$ be an object of $(\bigcup_{i\in N}q_o^+(i))\setminus (\bigcup_{i\in N}p^+(i))$. Assume that object $o_1$ belongs to $q_*^+(i_2)$ for a given agent $i_2$, and to $p^-(i_1)$ for another agent $i_1$. If $p^-(i_2)\neq \emptyset$ then starting from $p$, $i_1$ and $i_2$ could exchange $o_1$ with an object of $p^-(i_2)$ leading to an assignment $q$ where both (i) and (ii) hold. Otherwise, if $p^-(i_2)=\emptyset$, then let $o_2\in p^+(i_2)\setminus q_*^+(i_2)$ (which must exists since $i_2$ is not a clear winner). Assume that $o_2$ belongs to $q_*(i_3)$ for a given agent $i_3$. Note that $o_2$ must belong to $q_*^+(i_3)$, as otherwise $i_2$ and $i_3$ could exchange $o_1$ and $o_2$ in $q_*$ and we would obtain another assignment $q_{**}$ that still Pareto dominates $p$, but which is closer to $p$. Now, again, if $p^-(i_3)\neq \emptyset$ then starting from $p$ we could create a Pareto dominating assignment $q$ with properties (i) and (ii) by exchanging objects along this cycle, namely, by assigning $o_1$ to $i_2$, $o_2$ to $i_3$ and $o_3$ to $i_1$, where $o_3$ is an object of $p^-(i_3)$. However, if $p^-(i_3)= \emptyset$ then we continue the construction of the sequence.

The last case that we have to discuss is a possible repetition occurring in the above sequence. Suppose that for some indices $k<l$, $o_l\in q_*^+(i_k)$ for the first time in the sequence. Therefore, the agents involved in this sub-sequence exchange their top objects in $q_*$ compared to $p$. But then we can construct another assignment $q_{**}$ from $q_*$ by reassigning these objects to their original owners in $p$, contradicting with our assumption of $q_*$ being as close to $p$ as possible.
		\end{proof}

Based on the lemma, we obtain the following characterization of Pareto optimality in the bivalued case.

		\begin{lemma}\label{prop2util}
		An assignment $p$
		is Pareto dominated if and only if there exists an assignment $q$ such that (i) $\forall i\in N, |q^+(i)|\geq |p^+(i)|$ and (ii) $\exists j\in N, |q^+(j)|>|p^+(j)|$ and $p^-(j)\neq \emptyset$.
		\end{lemma}
		\begin{proof}
		One implication has already been proved in Lemma \ref{lemm1}. To prove the second implication we assume first that there exists $q$ such that (i) and (ii) hold. Let $j$ be the index described in (ii). For any $i\in N\backslash \{ j\}$, let $A_i$ be an arbitrary subset of $q^+(i)$ of size $|A_i|=|p^+(i)|$. Let $A_j$ be an arbitrary subset of $q^+(j)$ of size $|A_j|=|p^+(j)|+1$. Let $\bar{A}=O\setminus \bigcup_{i\in N}A_i$. Note that by definition $|\bar{A}|=|\bigcup_{i\in N}p^-(i)|-1$ because all objects are assigned, and therefore
$	|\bar{A}|=|O|-|\bigcup_{i\in N}A_i|=\sum_{i\in N}|p(i)|-\sum_{i\in N}|p^+(i)|-1
	=	\sum_{i\in N}|p^-(i)|-1.
$

		We arbitrarily partition $\bar{A}$ into $n$ subsets $\bar{A}_1, \ldots \bar{A}_n$ such that for all $i\in N\setminus \{ j\}, |\bar{A}_i|=|p^-(i)|$ and $|\bar{A}_j|=|p^-(j)|-1$.
Finally, let $q_*$ be the assignment such that for all $i\in N, q_*(i)=A_i\cup \bar{A}_i$.
By the construction of $q_*$, $|q_*^+(i)|\geq |p^+(i)|$ and $|q_*(i)|=|p(i)|$ hold for every $i\in N$, and $|q_*^+(j)|>|p^+(j)|$ and $|q_*(j)|=|p(j)|$ hold.
Therefore, $p$ is Pareto dominated by $q_*$.\end{proof}

If at least one object is not assigned then a trivial Pareto improvement would be to assign this object to an agent. Therefore, we can focus on the case where all objects are assigned. According to Theorem~\ref{prop2util}, a Pareto improvement can be computed by focusing on the assignment of top objects for the agents. We describe an algorithm based on maximum flow problems to obtain such an assignment. For any $i\in N$, let $G_i=(V_i, E_i)$ be a directed graph that models the search for a Pareto improvement for agent $i$ as a flow problem. Set of vertices $V_i$ contains one vertex per agent, one object per object, plus source $s$ and sink $t$. To ease the notation, vertices are denoted as the agents or the objects that they represent, and therefore, $V_i =N \cup O \cup \{ s, t\}$. The set of edges $E_i$ and their capacities are constructed as follow:
\begin{itemize}\setlength\itemsep{0em}
\item For any $l\in N$ and $o\in O$ such that $o\in E_l^1$, there is an arrow $(l, o)$ with capacity 1.
\item For any $o\in O$, there is an arrow $(o, t)$ with capacity 1.
\item For any $l\in N\setminus \{ i\}$, there is an arrow $(s, l)$ with capacity $|p^+(l)|$, and there is an edge $(s, i)$ with capacity $|p^+(i)|+1$.
\end{itemize}
It is easy to show that there exists a flow of value $\sum_{l\in N}|p^+(l)|+1$ if and only if there exists an assignment such that any agent $l\in N\setminus \{ i\}$ receives at least $|p(l)\cap E_l^1|$ top objects and agent $i$ receives $|p^+(i)|+1$ top objects. Hence, by Lemma~\ref{prop2util}, there exists a Pareto improvement of $p$ if and only if there exists $i\in N$ such that $p(i)\cap E_i^2\neq \emptyset$ and there exists a flow of value $\sum_{l\in N}|p^+(l)|+1$ in $G_i$. Therefore, finding a Pareto improvement can be performed in polynomial time by solving at most $n$ maximum flow problems.

\begin{theorem}\label{theo:bivalued}
Under bivalued utilities, there exists a polynomial-time algorithm for checking Pareto optimality and finding a Pareto improvement, if any.
\end{theorem}

Note that we can find a Pareto optimal assignment in polynomial time as well by repeatedly improve the assignment with Pareto improvements. In each Pareto improvement the number of top objects increases by at least one, and therefore there is at most $m$ Pareto improvements before reaching a Pareto optimal assignment.

\noindent
\begin{example}\label{example-two}
Let $n = 3$, $m = 6$, $p = (o_1 o_4 | o_2 o_5|o_3 o_6)$, $E_1^1 = \{o_1, o_2, o_3\}$, $E_2^1 = \{o_2\}$ and $E_3^1 = \{o_1, o_3, o_5, o_6\}$.  
$G_1$ is depicted in Figure~\ref{fig:flownetwork}.
The flow of value 5 (bold-lines) gives assignment $(o_1o_3|o_2o_4|o_5o_6)$, which Pareto dominates $p$.
\end{example}

\begin{figure}
	\begin{center}
	\includegraphics[clip,height=7cm,keepaspectratio]{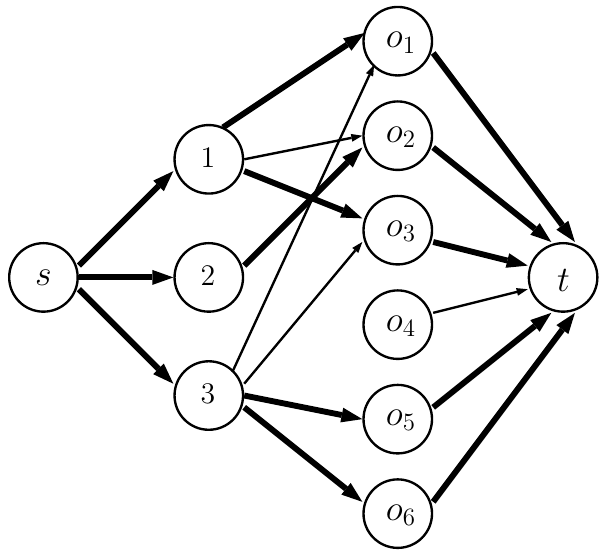}
	\end{center}
	\caption{Flow network $G_1$  in Example~\ref{example-two}.}
	\label{fig:flownetwork}
	\end{figure}

We end this section by a table (Table~\ref{table-summary}) summarizing our results for cardinal utilities.

\begin{table}[h!]
\begin{center}
	 \scalebox{0.8}{
\begin{tabular}{cccc}
\hline
\backslashbox{$\#$ agents}{utility values} & $\leq 2$ & lexicographic, $\geq 3$ & $\geq 3$ \\ \hline
$\geq$ 2, bounded & in {\sf P} & in {\sf P} & 
\begin{tabular}{l} weakly {\sf NP}-complete (Th.  \ref{th:lang})\\ pseudo-polynomial (Th.  \ref{theo:pseudopoly})\end{tabular} \\ \hline
unbounded & in {\sf P} (Th.  \ref{theo:bivalued}) & in {\sf P} (Th.  \ref{th:lexi-opt}) & \begin{tabular}{l} strongly {\sf NP}-complete  \\(Ths.  \ref{th:testpo-conpc} and \ref{th:testpo-three})\end{tabular}\\\hline
\end{tabular}
}
\end{center}
\label{table-summary}
\caption{Complexity of testing Pareto optimality}

\end{table}

\subsection{Conservative Pareto optimality}


In this section, we consider a related concept of conservative Pareto optimality and provide connections between testing Pareto optimality with testing conservative Pareto optimality. 
An assignment $p$ is \emph{conservatively Pareto optimal} if there does not exists another assignment $q$ that Pareto dominates $p$ and $|q(i)|=|p(i)|$ for all $i\in N$. 
Next we point out that a polynomial-time algorithm to test Pareto optimality can also be used to test conservative Pareto optimality in polynomial time. 

\begin{lemma}
	If there exists a polynomial-time algorithm to test Pareto optimality, then  there exists a polynomial-time algorithm to test conservative Pareto optimality.
	\end{lemma}
The argument is as follows. 
	For each $i\in N$ and $o\in O$, we update the utility $u_i(o)$ to utility $u_i(o)+C$ where $C$ is $m\times \max_{i\in N, o\in O}u_i(o)$. Note that by modifying the utility function, we ensure that an agent primarily cares about how many objects she gets and then cares about what are the particular utilities for the objects. Hence, for the modified utilities, in any Pareto improvement, no agent will get less objects than in her original allocation and hence each agent will get the same number of objects after reallocation. Hence if $q$ Pareto dominates $p$ with respect to the modified utilities if and only if then $q$ conservatively Pareto dominates $p$ with respect to the original utilities. Hence, by simply modifying the utilities as specified above, a polynomial-time algorithm to test Pareto optimality can be used to test conservative Pareto optimality.

For the converse, we note that as long as we allow zero utilities, if there exists a polynomial-time algorithm to test conservative Pareto optimality, then  there exists a polynomial-time algorithm to test Pareto optimality.

\begin{lemma}
	If there exists a polynomial-time algorithm to test conservative Pareto optimality, then  there exists a polynomial-time algorithm to test  Pareto optimality.
	\end{lemma}
	The argument is as follows. Consider an instance consisting of $n$ agents, $m$ objects, $u$ the utility matrix, and $p$ an assignment. We define the new set of objects by adding $(n-1)m$ dummy objects that each agent values at 0 (agent value the initial objects as in $u$). The assignment $q$ is defined from $p$ by giving $m - |p(i)|$ objects to agent $i$, for all $i$. Then $q$ is conservatively Pareto optimal for the modified instance if $p$ is Pareto optimal for the original instance.

\section{Ordinal preferences}\label{ordinal}

	
In this section, we consider the setting in which agents have additive cardinal utilities but only their ordinal preferences over objects are known by the central authority. This could be because the elicitation protocol did not ask the agents to communicate their utilities, or simply because they do not know them precisely.
We will assume in this section that for any underlying additive utility $u_i(o)>0$ for $i\in N$ and $o\in O$.
In this case, one can still reason whether a given assignment is Pareto optimal with respect to some or all cardinal utilities consistent with the ordinal preferences.
An assignment $p$ is \emph{possibly Pareto optimal} with respect to preference profile $\pref$ if there exists $u\in {\mathcal{U(\pref)}}$ such that $p$ is Pareto optimal for $u$. An assignment is \emph{necessarily Pareto optimal} with respect to preference profile $\pref$ if for any $u\in {\mathcal{U(\pref)}}$ the assignment $p$ is Pareto optimal for $u$.

\subsection{Possible Pareto Optimality}

We first note that necessary Pareto optimality implies possible Pareto optimality. Secondly, at least one necessarily Pareto optimal assignment exists in which all the objects are given to one agent.
Since the computation of a possibly or necessary Pareto optimal assignment can be performed in polynomial time, we focus on the problems of testing possible and necessary Pareto optimality.

				\begin{theorem}\label{th:posspo-equi} 
A (discrete) assignment is (1) possibly Pareto optimal if and only if (2) there exists no cycle in $G(p)$ which contains at least one edge corresponding to a strict preference if and only if (3) it is Pareto optimal under lexicographic utilities. 
					\end{theorem}
			\begin{proof}
				
We already established (2)$\iff$(3) in Theorem~\ref{th:lexi-opt}.
We also know that 
(3)$\implies $(1) because if an assignment that is Pareto optimal with respect to  lexicographic utilities, it is possibly Pareto optimal.
It remains to be shown that 
(1)$\implies $(2).
Suppose $p$ is not efficient with respect to lexicographic utilities, then by Theorem~\ref{th:lexi-opt}, $G(p)$ admits a cycle which contains at least one edge corresponding to a strict preference. Now  consider an assignment $q$ which is a result of exchange of objects along the cycle. If $o$ points to $o'$ in the cycle, then the agent getting $o$ in $p$ now gets $o'$ in place of $o$.  Since the cycle contains at least one one edge corresponding to a strict preference, assignment $q$ is a Pareto improvement over $p$ with respect to all utilities consistent with the ordinal preferences. Hence $p$ is not possibly Pareto optimal.
\end{proof}				

Since the characterization in Theorem~\ref{th:lexi-opt} also applies to possible Pareto optimality, hence possible Pareto optimality can be tested in linear time. The argument in the proof above also showed that
possible Pareto optimality is equivalent to Pareto optimality under lexicographic preferences.

The following example points out that a possibly Pareto optimal assignment may not be a necessarily Pareto optimal assignment.	
\begin{example}\label{example-npo}
Consider 
 two agents with identical preferences $o_1\succ o_2 \succ o_3 \succ o_4$. 
 Every assignment is possibly Pareto optimal; however assignment $p=(o_1o_4|o_2o_3)$ is not necessarily Pareto optimal since it is not Pareto optimal for the following utilities.
$$
\begin{array}{c|ccccc}
& o_1 & o_2 & o_3 & o_4 \\ \hline
1& 10 &9 &8 &7 \\
2 & 10 & 3 & 2 & 1\\
\end{array}
$$
\end{example}

\subsection{Necessary Pareto Optimality}

Next we present two characterizations of necessary Pareto optimality.
	The first highlights that necessary Pareto optimality is identical to the necessary Pareto optimality considered by \citet{BEL10a}.

We first define the 
\emph{responsive set extension}~\cite{barbera2004ranking}, which 
extends 
preferences over objects to 
preferences over sets of objects.
Formally,
for agent $i\in N$, her preferences
$\pref_i$ over $O$ are extended to her preferences $\pref_i^{RS}$ over $2^O$
as follows:
 $q(i) \mathrel{\pref_i^{RS}} p(i)$ if and only if there exists an injection $f$ from $p(i)$ to $q(i)$ such that for each $o\in p(i)$, $f(o)\pref_i o$.
 Since $\pref_i^{RS}$ is a partial order, we say a preference $R_i$ is a \emph{completion} of $\pref_i^{RS}$ if it is a complete and transitive relation over sets of objects that is consistent with $\pref_i^{RS}$.

 We say that an assignment is  \emph{RS-efficient} if it is Pareto optimal with respect to the RS set extension relation of the agents. An  assignment $p$ is RS-efficient, if there does not exist another assignment $q$ such that $q(i)\pref_i^{RS} p(i)$ for all $i\in N$ and $q(i)\succ_i^{RS} p(i)$ for some $i\in N$. We say that $q$ strictly RS-dominates $p$ if $q$ Pareto dominates $p$ with respect to RS.

\begin{theorem}\label{th:car-ext}
		An assignment is necessarily Pareto optimal if and only if it is Pareto optimal under all completions of the responsive set extension.
		\end{theorem}
		\begin{proof}
		If an assignment is not Pareto optimal under certain additive preferences, it is by definition not Pareto optimal under this particular completion of responsive preferences.

			Assume that assignment $p$ is not Pareto optimal under some completion of the responsive set extension. Then there exists another assignment $q$ in which
			for all $i\in N$
			$q(i) \succsim_i^{RS} p(i)$ or ($(p(i)\nsucc_i^{RS} q(i)$ and $q(i)\nsucc_i^{RS} p(i))$),
			and for some $j\in N$, $q(j) \succ_j^{RS} p(j)$ or ($p(j)\nsucc_j^{RS} q(j)$ and $q(j)\nsucc_j^{RS} p(j)$)
		For both cases, if the allocations are incomparable with respect to responsive set extension, then there exists an object $o$ such that $|q(i)\cap \{o'\midd o\pref_i o\}|>|p(i)\cap \{o'\midd o\pref_i o\}|$. In that case, consider a utility function $u_i$ in which $u_i(o''')-u_i(o'')\leq \epsilon$ for all $o''',o''\pref_i o$ and $u_i(o)>\sum_{o'\prec_i o}u_i(o')+|O|\epsilon$. For $u_i$, $u_i(q(i))>u_i(p(i))$.
		\end{proof}
		
%

For our second characterization of necessarily Pareto optimal assignments, we define a
 \emph{one-for-two Pareto improvement swap} as an exchange between two agents $i_j$ and $i_k$ involving 
 objects $o_j^1$, $o_j^2\in p(i_j)$ and $o_k\in p(i_k)$ such that $o_k\succ_{i_j}o_j^1\pref_{i_j}o_j^2$.

\begin{theorem}\label{thm:char_nec_PO}
An assignment $p$ is necessarily Pareto optimal if and only if
\begin{enumerate}
\item[(i)] it is possibly Pareto optimal and
\item[(ii)] it does not admit a one-for-two Pareto improvement swap.
\end{enumerate}
\end{theorem}

\begin{proof}
	We first show that if an assignment does not satisfy the two conditions, then it is not necessarily Pareto optimal.
Possible Pareto optimality is a requirement for the assignment to be necessarily Pareto optimal. To see that the second condition is also necessary, we have to show that if $p$ admits a one-for-two Pareto improvement swap then $p$ is not necessarily Pareto optimal. This is because the swap could indeed be a Pareto improvement for these two agents with the following utilities: $u_{i_j}(o_k)>2u_{i_j}(o_j^1)(\geq u_{i_j}(o_j^1)+u_{i_j}(o_j^2))$ and $u_{i_k}(o_k)<u_{i_k}(o_j^1)+u_{i_k}(o_j^2)$. These utilities are compatible with the ordinal preferences of these agents, because of the assumption $o_k\succ_{i_j}o_j^1\pref_{i_j}o_j^2$ (and irrespective to the ordinal preferences of $i_k$).

\smallskip

										Conversely, to show that conditions (i) and (ii) are sufficient for the assignment to be necessarily Pareto optimal, suppose for a contradiction that (1) $p$ is not necessarily Pareto optimal and (2) $p$ does not admit a one-for-two Pareto improvement swap. We will then show that there is an assignment that strictly RS-dominates $p$, implying that $p$ cannot be possibly Pareto optimal.
										
										From (1) and Theorem \ref{th:car-ext}, we have (3) there is another assignment $q$ and a collection of additive utility functions $u = (u_1, \ldots, u_n) \in \mathcal{U}(\pref)$ such that $q$ Pareto dominates $p$ with respect to $u$.

										Without loss of generality 
										  we may assume that each agent receives a nonempty bundle in $p$.
Regarding \textbf{the structure} of $p$, first we observe that the lack of one-for-two Pareto improvement swaps implies that every agent is assigned to some (or none) of her top objects and possibly to one additional object that she ranks lower. Formally, let $T_p(i)$ denote a set of $i$'s top objects she is assigned to in $p$, i.e., $T_p(i)=\{o: o\in p(i) \mbox{ s.t. } \nexists o'\notin p(i), o'\succ_i o\}$. Then $p(i)=T_p(i)\cup \{w_p(i)\}$, where $w_p(i)$ is either a single object or no object.

										We show that $|q(i)|=|p(i)|$ must hold for every agent $i$. Suppose not, then there is an agent $i$ for which $|q(i)|<|p(i)|$. By the definition of $T_p(i)$ it is straightforward that if $w_p(i)=\emptyset$ then $u_i(p(i))=u_i(T_p(i))> u_i(q(i))$, and if $w_p(i)\neq \emptyset$ then $u_i(p(i))=u_i(T_p(i)\cup \{w_p(i)\})> u_i(T_p(i))\geq u_i(q(i))$, a contradiction.
										Furthermore, for every agent $i$, if $\{w_p(i)\}\neq\emptyset$ then for any object $o\in q(i)$ we have $o \succsim_i w_p(i)$.
										Otherwise, if there was an agent $i$ with $o\in q(i)$ such that $w_p(i)\succ_i o$, then $u_i(T_p(i))\geq u_i(q(i)\setminus\{o\})$ would imply $u_i(p(i))=u_i(T_p(i)\cup \{w_p(i)\}) > u_i(q(i))$.

										Now we construct a so-called \textbf{Pareto improvement sequence} with respect to $p$ and $q$, which consists of a sequence of agents $\{i_1, i_2, \dots i_k\}$ with possible repetitions and a set of distinct objects $\{o_1, o_2, \dots , o_m\}$ such that
										\begin{itemize}
										\item $o_1 \in q(i_2) \setminus p(i_2)$, $o_2 \in p(i_2) \setminus q(i_2)$, and $o_1 \succsim_{i_2} o_2$;
										\item $o_2 \in q(i_3) \setminus p(i_3)$, $o_3 \in p(i_3) \setminus q(i_3)$, and $o_2 \succsim_{i_3} o_3$;
										\item \ldots
										\item $o_m \in q(i_1)\setminus p(i_1)$, $o_1\in p(i_1)\setminus q(i_1)$, and $o_{m}\succsim_{i_1}o_1$.
										\end{itemize}
										and with strict preference for at least one agent.

										The presence of the above Pareto improvement sequence would imply the existence of an assignment $q'$ that RS-dominates $p$, 
										obtained by letting the agents exchange their objects along the sequence, i.e., with $q'(i)=p(i)\cup \{o_{k-1}:i_k=i, k=1,\ldots ,m\}\setminus \{o_k:i_k=i, k=1,\ldots,m\}$. This would contradict our assumption that $p$ is possibly Pareto optimal.

										We first define three types of agents, and a \textbf{one-to-one mapping $\pi$} from a subset of $O$ to itself such that if $o\in p(i)\setminus q(i)$ and $\pi(o)\in q(i)\setminus p(i)$ then $i$ is indifferent between these two objects. In the set $X$ we put all the agents with either no $w_p(i)$ or with $w_p(i)\in q(i)$. Each agent $i$ in this set must be indifferent between all objects in $(p(i)\setminus q(i))\cup (q(i)\setminus p(i))$ (i.e., these object are in a single tie in $i$'s preference list) by the following reasons. $|p(i)|=|q(i)|$ implies $|p(i)\setminus q(i)|=|q(i)\setminus p(i)|$. By the definition of $T_p(i)$ it follows that any object in $p(i)\setminus q(i)$ is weakly preferred to any object in $q(i)\setminus p(i)$ by $i$. However, from (3) we have $u_i(q(i))\geq u_i(p(i))$, which implies that $u_i(q(i)\setminus p(i))\geq u_i(p(i)\setminus q(i))$, which can only happen if $i$ is indifferent between any two objects in $(p(i)\setminus q(i))\cup (q(i)\setminus p(i))$. Let $\pi$ map $q(i)\setminus p(i)$ to $p(i)\setminus q(i)$ as a bijective function.

										Next, let $Y$ contain every agent $i$ who has object $w_p(i)$ such that there is an object $o\in q(i)\setminus p(i)$ with $o\sim_i w_p(i)$. 
										In this case $i$ must be indifferent between all objects in $(T_p(i)\setminus (q(i)\setminus \{o\}))\cup ((q(i)\setminus \{o\})\setminus T_p(i))$.

										Indeed, $|p(i)|=|q(i)|$ implies $|T_p(i)\setminus (q(i)\setminus \{o\})|=|(q(i)\setminus \{o\})\setminus T_p(i)|$.

										By the definition of $T_p(i)$ any object in $T_p(i)\setminus (q(i)\setminus \{o\})$ is weakly preferred to any object in $(q(i)\setminus \{o\})\setminus T_p(i)$ by $i$. On the other hand, $u_i(q(i))\geq u_i(p(i))$ and $o\sim_i w_p(i)$ implies $u_i((q(i)\setminus \{o\})\setminus T_p(i))\geq u_i(T_p(i)\setminus (q(i)\setminus \{o\}))$, leading to the conclusion that $i$ must be indifferent between all objects in $(T_p(i)\setminus (q(i)\setminus \{o\}))\cup ((q(i)\setminus \{o\})\setminus T_p(i))$. Therefore $\pi$ can map $o$ to $w_p(i)$ and $(q(i)\setminus \{o\})\setminus T_p(i)$ to $T_p(i)\setminus q(i)\setminus \{o\})$.

										Thirdly, let $Z$ contain every agent $i$ with object $w_p(i)$ such that for every $o\in q(i)$, $o\succ_i w_p(i)$. Note that there is at least one agent in $Z$, the one who gets strictly better off in $q$, as otherwise, if there was an object $o\in q(i)$ such that $w_p(i)\pref_i o$, then $u_i(T_p(i))\geq u_i(q(i)\setminus \{o\})$ would imply $u_i(p(i))=u_i(T_p(i)\cup \{w_p(i)\})\geq u_i(q(i))$.

										Finally, we shall note that if $T_p(i)$ is empty then $|p(i)|=|q(i)|=1$, so either $i$ is indifferent between $p(i)=\{w_p(i)\}$ and $q(i)$, in which case $i$ is in $Y$ with $\pi(q(i))=p(i)$, or $i$ strictly prefers $q(i)$ to $p(i)$ and then $i$ belongs to $Z$.

										To summarize, so far we have that for any $i\in X\cup Y$ and $o\in q(i)\setminus p(i)$ we associate an object $\pi(o)\in p(i)\setminus q(i)$ such that $o\sim_i \pi(o)$. Furthermore, for any $i\in Z$ and $o\in q(i)\setminus p(i)$ we have that $o\succ_i w_p(i)$.

										We build a Pareto improvement sequence as a part of a sequence involving agents $i_1, i_2, \dots$ with corresponding objects $o_1, o_2, \dots$ starting from any $i_1\in Z$ with $o_1=w_p(i)$. For every $k\geq 2$, let $i_k$ be the agent who receives 
$o_{k-1}$ in $q$. If $i_k\in X\cup Y$ then let $o_k=\pi(o_{k-1})$, and if $i_k\in Z$ then let $o_k=w_p(i)$. We terminate the sequence when an object is first repeated. 
This repetition must occur at some agent in $Z$, since for any agent $i$ the objects in $q(i)\setminus p(i)$ are in a one-to-one correspondence with those in $p(i)\setminus q(i)$ by 
$\pi$.

										Let the first repeated object belong to, say, $i_s=i_t\in Z$ for indices $1\leq s<t$. We show that the sequence $i_{s}, \dots, i_{t-1}$ is a Pareto improvement sequence. To see this, let us first consider an agent $i\in X\cup Y$. Whenever $i$ appears in the sequence as $i_k\in \{i_{s+1}, \dots, i_t\}$ she receives object $o_{k-1}\in q(i)\setminus p(i)$ and in return she gives away $\pi(o_{k-1})=o_k\in p(i)\setminus q(i)$, where $i$ is indifferent between $o_{k-1}$ and $o_k$. Now, let $i\in Z\setminus \{i_t\}$ that appears as $i_l\in \{i_{s+1}, \dots, i_t\}$. She receives object $o_{l-1}\in q(i)\setminus p(i)$ and in return she gives away  $w_p(i)=o_k\in p(i)\setminus q(i)$, where $o_{l-1}\succ_i w_p(i)$ by the definition of $Z$. Since $i$ appears in this sequence only once, it is obvious that $u_i(q(i)) > u_i(p(i))$. Finally, regarding $i=i_s=i_t\in Z$, $i$ receives $o_{t-1}\in q(i)\setminus p(i)$ and she gives away  $w_p(i)=o_{s}\in p(i)\setminus q(i)$, where $o_{t-1}\succ_i w_p(i)$. So we constructed a Pareto improvement sequence, and therefore $p$ is not possibly Pareto optimal, a contradiction.
\end{proof}

In Example~\ref{example-npo}, 
$p$ is not necessarily Pareto optimal because it admits a one-for-two Pareto improvement swap:
$o_2$, $o_3\in p(2)$, $o_1\in p(1)$ and $o_1\succ_{2}o_2\pref_{2}o_3$.
Example~\ref{example-npo} also shows that although an assignment may not be necessarily Pareto optimal there may not be any assignment that Pareto dominates it for \emph{all} utilities consistent with the ordinal preferences.
The characterization above also gives us a polynomial-time algorithm to test necessary Pareto optimality.

Finally, we would like to highlight that in some natural situation, where the preferences of the agents are aligned, all necessary Pareto optimal assignment can be seen as very unfair, as we illustrate in the following example.

\begin{example}\label{example-necessary}
Suppose that we have $m$ objects and $n$ agents, where $n<<m$ and all agents have the same preferences, e.g.\ they order objects according to their indices. In this situation in any necessarily Pareto efficient assignment essentially one agent must get all the $m-n$ top objects and at most one less preferred one, whilst every other agent may receive at most one object. So, if agent $1$ is the lucky one and $p$ is a necessarily Pareto optimal assignment then $p(1)=\{o_1,o_2,\dots , o_{m-n}\}$ must hold, with the possibility of agent 1 getting a further extended set of the top objects, whilst each of the other agents may hold at most one object from the rest of the objects (including agent 1).
\end{example}

\section{Conclusions}

We have studied, from a computational point of view, Pareto optimality in resource allocation under additive utilities and ordinal preferences. 
Many of our positive algorithmic results come with characterizations of Pareto optimality that improve our understanding of the concept and may be of independent interest.

Future work includes identifying restrictions on the preferences under which Pareto optimal and individually rational reallocation can be done in a computationally efficient manner.

\section*{Acknowledgements}  
This paper is an extended version of a paper that was presented at 
AAMAS (International Conference on Autonomous Agents \& Multiagent Systems) 2016~\citep{ABL+16a} and COMSOC (International Workshop on Computational Social Choice) 2016. 
Haris Aziz is supported by a Julius Career Award. P\'eter Bir\'o acknowledges support from the Hungarian Academy of Sciences under its Momentum Programme (LP2016-3/2016), the Hungarian Scientific Research Fund, OTKA, Grant No. K108673, and the J\'anos Bolyai Research Scholarship of the Hungarian Academy of Sciences. J\'er\^ome Lang, Julien Lesca and J\'er\^ome Monnot acknowledge support from the ANR project CoCoRICo-CoDec.  Part of the work was conducted when P\'eter Bir\'o visited Paris Dauphine and when Julien Lesca visited Corvinus University of Budapest sponsored by COST Action IC1205 on Computational Social Choice.

\bibliographystyle{plainnat}

\end{document}